\documentclass[12pt]{iopart}
\def\version{September 12, 2013}

\usepackage{bm}
\usepackage{upgreek}

\usepackage{slashed}
\usepackage{latexsym}
\usepackage{amsthm}
\usepackage{amssymb}

\usepackage{amsmath}

\DeclareOldFontCommand{\brianup}{\upshape}{\mathrm}

\DeclareSymbolFont{EUR}{U}{eur}{m}{n}
\SetSymbolFont{EUR}{bold}{U}{eur}{b}{n}
\DeclareSymbolFontAlphabet{\eur}{EUR}

\DeclareSymbolFont{EUB}{U}{eur}{b}{n}
\SetSymbolFont{EUB}{bold}{U}{eur}{b}{n}
\DeclareSymbolFontAlphabet{\eub}{EUB}

\DeclareSymbolFont{AMSb}{U}{msb}{m}{n}
\DeclareSymbolFontAlphabet{\mathbb}{AMSb}

\newcommand\charge{ \mathrm{e} }

\newcommand\eubD{\eub{D}}
\newcommand\eubJ{\eub{J}}
\newcommand\eubL{\eub{L}}

\newcommand\bmupalpha{\bm\upalpha}
\newcommand\bmupbeta{\bm\upbeta}

\newcommand\bmupphi{\bm\upphi}

\newcommand\bmupxi{\bm\upxi}

\newcommand{\notyet}[1]{{}}

\newcommand{\p}{\partial}
\renewcommand{\P}{\grave{\partial}}

\def\R{\mathbb{R}}
\newcommand{\C}{\mathbb{C}}

\newcommand{\N}{\mathbb{N}}\newcommand{\Z}{\mathbb{Z}}

\newcommand{\abs}[1]{\vert #1 \vert}

\newcommand{\norm}[1]{\Vert #1 \Vert}

\newcommand{\sothat}{{\rm ;}\ }

\DeclareMathSymbol{\varGamma}{\mathord}{letters}{"00}
\DeclareMathSymbol{\varDelta}{\mathord}{letters}{"01}
\DeclareMathSymbol{\varTheta}{\mathord}{letters}{"02}
\DeclareMathSymbol{\varLambda}{\mathord}{letters}{"03}
\DeclareMathSymbol{\varXi}{\mathord}{letters}{"04}
\DeclareMathSymbol{\varPi}{\mathord}{letters}{"05}
\DeclareMathSymbol{\varSigma}{\mathord}{letters}{"06}
\DeclareMathSymbol{\varUpsilon}{\mathord}{letters}{"07}
\DeclareMathSymbol{\varPhi}{\mathord}{letters}{"08}
\DeclareMathSymbol{\varPsi}{\mathord}{letters}{"09}
\DeclareMathSymbol{\varOmega}{\mathord}{letters}{"0A}

\theoremstyle{plain}
\newtheorem{lemma}{Lemma}[section]

\newtheorem{theorem}[lemma]{Theorem}

\theoremstyle{definition}

\theoremstyle{remark}
\newtheorem{remark}[lemma]{Remark}

\renewcommand{\theequation}{\thesection.\arabic{equation}}

\makeatletter\@addtoreset{equation}{section}
\makeatother

\def\spec{\sigma}

\renewcommand{\Re}{\mathop{\rm{R\hskip -1pt e}}\nolimits}
\renewcommand{\Im}{\mathop{\rm{I\hskip -1pt m}}\nolimits}

\begin{document}

\title{Polarons as stable solitary wave solutions to the Dirac--Coulomb system}

\author{
{\sc Andrew Comech}
\\
{\small\it Texas A\&M University, College Station, TX 77843, U.S.A.}
\\
{\small\it Institute for Information Transmission Problems, Moscow 101447, Russia}
\\
{\sc Mikhail Zubkov}
\\
{\small \it University of Western Ontario,  London, ON N6A 5B7, Canada}
\\
{\small\it Institute for Theoretical and Experimental Physics, Moscow 117259, Russia}
}

\date{\version}

\begin{abstract}
We consider solitary wave solutions to the Dirac--Coulomb system
both from physical and mathematical points of view.
Fermions interacting with gravity in the Newtonian limit
are described by the model of Dirac fermions with the Coulomb attraction.
This model also appears in certain condensed matter
systems with emergent Dirac fermions interacting via optical phonons.
In this model, the classical soliton solutions
of equations of motion describe the physical objects
that may be called polarons,
in analogy to the  solutions of the Choquard equation.
We develop analytical methods
for the Dirac--Coulomb system,
showing that the no-node gap solitons
for sufficiently small values of charge are linearly (spectrally) stable.
\end{abstract}


\maketitle

\section{Introduction}
\label{sectintr}

The Dirac--Maxwell system and other models of fermion fields with self-interaction
(such as the massive Thirring model \cite{MR0091788} and the Soler
model \cite{PhysRevD.1.2766}) have been attracting the interest of both
physicists and mathematicians for many years.
These models,
just like the nonlinear Schr\"odinger equation,
have localized solitary wave solutions
of the form
$\phi(x)e^{-i\omega t}$,
where $\phi(x)$ is exponentially localized in space.
For the Dirac--Maxwell system the
localized solutions
with $\omega\in(-m,m)$
have been shown to exist (first numerically and then analytically)
in
\cite{wakano-1966,MR1364144,MR1386737,MR1618672}.
These solutions may encode certain properties of the theory
which can not be obtained via perturbative analysis.
The role of these classical solutions
in high energy physics
has long been the topic of intense discussion (see, for example,
\cite{PhysRevD.10.517,PhysRevD.10.3235,PhysRevD.12.3880,PhysRevD.14.535,
PhysRevD.29.985,PhysRevD.31.2701}).
It seems that there is no physical meaning
of such solitary waves in quantum electrodynamics.
These classical localized states are formed
due to the attraction of the spinor field to itself,
which takes place
at the energies $\omega\gtrsim -m$.
On the other hand,
the associated quantum field theory
admits the appearance of antiparticles.
To describe antiparticles, it is necessary to change the order
of the fermion creation-annihilation operators.
This results in the additional change of sign at the scalar potential.
Thus, in the quantum theory for $\omega\gtrsim -m$
instead of the self-attraction (which could lead to the
creation of a localized mode)
one again ends up with the self-repulsion in the antiparticle sector.
This anticommuting nature of fermion variables
is ignored in the classical Dirac--Maxwell system.

\medskip

In the present paper we make an attempt
to determine how the solitary waves
of the classical equations of motion
could play a role in the quantum field theory.
We show that the Dirac equation with the Coulomb attraction
emerges in the semiclassical description of fermions
interacting with optical phonons or with the gravitational field.

Let us give more specifics.
It is well known that the relativistic Dirac fermions may emerge in the
condensed matter systems.
This occurs, for example,
on the boundary of the 3D topological insulators
and in some two-dimensional structures like graphene
\cite{nature-438-197,2011arXiv1111.3017F}.
Moreover, massless Dirac fermions appear in
3D materials
(see, for example, \cite{2012PhRvL.108n0405Y} and references
therein) at the phase transition between a 
topological and a normal insulator. Recently, the existence of massless
fermions at the phase transition between the insulator states with different
values of topological invariants has been proven
for a wide class of relativistic models \cite{MR2911335}.
When the interaction with other fields
is taken into account,
these massless fermions gain the mass.
Fermion excitations in various materials
interact with phonons \cite{MR766230}.
As a result, the attractive interaction between the fermions appears. Such an interaction
gives rise to the formation of Cooper pairs
in microscopic theories of superconductivity \cite{MR766230,schmidt-1997}.
In the case of the exchange by optical phonons
\cite{1990AnPhy.202..320C} the interaction has the form of the Coulomb attraction.
This leads to the formation of the polaron,
as in the nonrelativistic Landau--Pekar approach \cite{polaron-1933,polaron-1946}.

While in the above model
the Lorentz symmetry is broken due to the Coulomb forces, 
the relativistically-invariant version of a similar system
is given by the Dirac fermions interacting with the gravitational field.	
This problem may be thought of as a true relativistic polaron
problem.
We show that
in the Newtonian limit we again arrive at the system of Dirac fermions
interacting via Coulomb-like forces.
Such polarons may emerge in the unified theories.

It is worth mentioning that Positronium
(the bound state of electron-positron pair)
has nothing to do with the solitary waves discussed in the present paper.
Positronium can be described
by the Dirac equation in the external Coulomb field.
This is in contrast to our solitary waves,
which are described by the Dirac equation
in the potential created by the spinor field itself.

\medskip

In the second part of the paper
we analyze the stability of solitary wave solutions
in the Dirac--Coulomb system.
We show that certain solitary waves are \emph{linearly stable}, i.e.
the spectrum of the equation linearized at a particular solitary wave
has no eigenvalues with positive real part.
Our approach is based on the fact that the nonrelativistic limit
of the Dirac--Coulomb system is the Choquard equation.
In particular, the solitary wave solutions
to the Dirac--Coulomb system are obtained as a bifurcation
from the solitary waves of the Choquard equation.
(It is worth mentioning that the Choquard equation
appears in the conventional
polaron problem \cite{polaron-1933,polaron-1946}.)
It also follows that
the eigenvalue families of
the Dirac--Coulomb system linearized at a solitary wave
are deformations of eigenvalue families
corresponding to the Choquard equation
(this has been rigorously proved in \cite{dirac-spectrum}
in the context of nonlinear Dirac equations).
The latter could be analyzed
via the Vakhitov--Kolokolov stability criterion \cite{VaKo}.
The delicate part is the absence of bifurcations
of eigenvalues from the continuous spectrum.
We explain that this follows from the limiting absorption principle
for the free Dirac operator.

We emphasize that
we present one of the first results
for the linear stability of spatially localized fermion modes.
Prior attempts at their stability properties
included the analysis of
stability of Dirac solitary waves with respect to
particular families of perturbations
(such as dilations),
see e.g. 
\cite{MR592382,PhysRevD.36.2422,PhysRevE.82.036604,MR848095}
and related numerical results
\cite{MR637021,MR710348,PhysRevLett.50.1230,chugunova-thesis}.
Yet neither the linear stability
nor orbital nor asymptotic stability
were understood.
Our latest results
on linear stability and instability
for the nonlinear Dirac equation are in
\cite{
MR2892774,dirac-nd}.
There are also recent results
on asymptotic stability
of solitary wave solutions
to the nonlinear Dirac equation
in 1D and in 3D
\cite{MR2985264,MR2924465}
(proved under the assumption that a particular solitary wave is linearly stable).
Interestingly, the orbital stability
has been proven
for small amplitude solitary waves
in the completely integrable massive Thirring model \cite{1304.1748v1_2}.
It should be mentioned that in contrast to those of the nonlinear Dirac equation, 
the questions of linear, orbital,
and even asymptotic stability of solitary waves
in the nonlinear Schr\"odinger equation
are essentially settled
(see e.g. \cite{VaKo,MR901236,MR1334139,MR2231327}).

We would also like to mention that		
Einstein-Dirac equations		
were shown to have		
particle-like solutions
which are linearly stable
with respect to spherically symmetric perturbations
\cite{PhysRevD.59.104020}.

\medskip

The paper is organized as follows.
In Section~\ref{sectdirac}, we briefly describe the
model under investigation.
In Section~\ref{sectquant} we
discuss the appearance of the considered solitary waves
in a system of Dirac fermions
interacting with optical phonons.
In Section~\ref{sectgrav} we
consider the relation of solitary waves to the gravitational polaron.
In Section~\ref{sect-stability}
we sketch the proof of existence of solitary waves
and then address the question of their linear stability.
Our conclusions are in Section~\ref{sectconcl}.	
In the Appendix we give the details		
of the Vakhitov--Kolokolov stability criterion \cite{VaKo}
in its application to the Choquard equation.

\section{Dirac--Coulomb system}
\label{sectdirac}

Below, we choose the units so that
$\hbar=c=1$.
We consider the model with the action
\begin{equation}\label{gact}
S_E = \int\!\!d^3x\,dt\,\bar\zeta(i D\sb{\varphi}-m)\zeta
-
\int\!\!d^3x\,dt\,
\frac{(\bm\nabla\varphi)^2}{2}
,
\end{equation}
where the Dirac fermion field
interacts with itself via an instantaneous Coulomb
interaction,
$x\in\R^3$, $t\in\R$,
$\charge^2$ is the coupling constant, $\zeta(x,t)\in\C^4$
is a four-component Dirac field,
$\varphi(x,t)\in\R$ is the Coulomb field,
and
$$
D\sb{\varphi} =\gamma^0(\p_0+i \charge\varphi)+\sum\sb{j=1}\sp{3}\gamma\sp j\p\sb j,
$$
where
$\p\sb 0=\frac{\p}{\p t}$,
$\p\sb j=\frac{\p}{\p x\sp j}$, $1\le j\le 3$.
The Dirac matrices $\gamma\sp\mu$, $0\le\mu\le 3$, satisfy the
Euclidean--Clifford algebra $\{\gamma\sp\mu,\gamma\sp\nu\}=2 g\sp{\mu\nu}$,
with
$g\sp{\mu\nu}$
the inverse of the metric tensor
$g\sb{\mu\nu}=\mathop{\rm diag}[1,-1,-1,-1]$.
Above, $\charge$ is the charge of the spinor field,
$\varphi(x,t)$ is the (real-valued) external scalar field
(such as the potential of the electric field in $\R^3$).
We denote
$(\bm\nabla\varphi)^2=\sum\sb{j=1}\sp{3}(\p\sb j\varphi)^2$.
It is worth
mentioning that in this model the Lorentz symmetry is broken due to
the Coulomb forces.

The dynamical equations
corresponding to \eqref{gact}
are given by the following
Dirac--Coulomb system:
\begin{equation}\label{2dd-0}
\left\{
\begin{array}{l}
i\p\sb{t}\zeta
=-i\bm\alpha\!\cdot\!\bm\nabla
\zeta+m\beta\zeta
+\charge\varphi\zeta,
\\
\Delta \varphi=\charge\zeta\sp\ast\zeta,
\end{array}
\right.
\end{equation}
where
$\zeta(x,t)\in\C^{4}$,
$\varphi(x,t)\in\R$,
$x\in\R^3$,
and $\Delta=\sum\sb{j=1}\sp{3}\p\sb{x\sb j}^2$.
Above,
$\bm\alpha=(\alpha\sp 1,\alpha\sp 2,\alpha\sp 3)$,
where the self-adjoint Dirac matrices
$\alpha\sp j$ and $\beta$
are related to $\gamma\sp\mu$
by
$$
\gamma\sp j=\gamma\sp 0\alpha\sp j,
\quad
1\le j\le 3;
\qquad
\gamma\sp 0=\beta.
$$
They satisfy
$
(\alpha\sp j)^2=\beta^2=I\sb{4},
$
$
\alpha\sp j\alpha\sp k+\alpha\sp k\alpha\sp j=2 I\sb{4}\delta\sb{jk},
$
$
\alpha\sp j\beta+\beta\alpha\sp j=0;
$
$1\le j,k\le 3$.
Above, $\bar\zeta=(\beta\zeta)\sp\ast=\zeta\sp\ast\beta$,
with
$\zeta\sp\ast$
the hermitian conjugate of $\zeta$.
A particular choice of the Dirac matrices
does not matter;
we take the Dirac matrices in the common form
\begin{equation}
\alpha\sp j
=\begin{pmatrix}0&\sigma_j\\\sigma_j&0\end{pmatrix},
\qquad
\beta
=\begin{pmatrix}I_2&0\\0&-I_2\end{pmatrix},
\end{equation}
where $I\sb{2}$ is the $2\times 2$ unit matrix
and
$\sigma\sb{1}=\left(\begin{matrix} 0&1 \\ 1&0\end{matrix}\right)$,
$\sigma\sb{2}=\left(\begin{matrix} 0&-i \\ i&0\end{matrix}\right)$,
$\sigma\sb{3}=\left(\begin{matrix} 1&0 \\ 0&-1\end{matrix}\right)$
are the Pauli matrices.

\begin{remark}\label{remar-signs}
Note that according to the second equation in \eqref{2dd-0},
if $\varphi\to 0$ at infinity, then
$\charge\varphi$ is strictly negative
and behaves like an attractive Coulomb potential
in the first equation in \eqref{2dd-0}
(for energies near $m$),
leading to the existence of bound states
for $\omega\lesssim m$.
\end{remark}

\medskip

\section{Polarons
due to interaction of fermion field with optical phonons}
\label{sectquant}


\subsection{Field-theoretical description of the generalized Frohlich model}

It was mentioned in the introduction that Dirac fermions may emerge in various
3D systems at the phase transition between the insulating states with
different values of momentum space topological invariants. Similarly to the
ordinary electrons in crystals, these Dirac fermions may interact with optical
phonons,
thus giving rise to the polaron problem. The model Hamiltonian
for a Dirac particle interacting with optical phonons can be obtained via the
generalization of the conventional Frohlich Hamiltonian
\cite{froelich-1954, polaron_lectures, 1990AnPhy.202..320C, polaron-1933, polaron-1946}:
\begin{equation}
{\cal H}=
\sum\sb{j=1}\sp{3}
\sum_p c^+_p [\gamma^0 \gamma^j \hat{\bf p}_j+ m \gamma^0]c_p
+
\sum_{k,p} [i \tilde{\alpha} c^+_{p+k}\frac{1}{|{\bf k}|} \hat{a}_k c_p + H.C.]
+ \Omega \sum_k \hat{a}^+_k \hat{a}_k.
\label{frohlich}
\end{equation}
Here $\tilde{\alpha}$ and $\Omega$ are coupling constants, $c^+_k$ are
the electron creation operators,
and $\hat{a}^+_k$ are the phonon creation operators.
We introduce the phonon field
$\varphi(x) = i \sqrt{2\Omega} \sum_k \frac{1}{|k|} \hat{a}_k
e^{ikx}+H.C.$ and the electron field $\psi(x) = \sum_p c_p e^{i p x}$.

We express ${\rm Tr}\, e^{-i {\cal H} T}$, where $T$ is time, as a functional
integral.
This is done as follows. First, we subdivide the time interval $[0,T]$
into
smaller intervals $\Delta T$ and represent $e^{-i {\cal H} T} = e^{-i {\cal H}
\Delta T}\dots e^{-i {\cal H} \Delta T}$. Next, we substitute
$$
1 = \frac{1}{2\pi
i}\int\!\!d \eta\,d \bar{\eta}\,d \zeta\,d\bar{\zeta}\,e^{-\bar{\zeta}\zeta -
\bar{\eta}\eta} \,|\eta\,\zeta\rangle \langle \eta\,\zeta|,
$$
where the coherent
states are
$|\eta\,\zeta\rangle = e^{\eta c^+ + \zeta \hat{a}^+}|0\rangle$, $\zeta\in\C$,
and $\eta$ is the Grassmann variable.
The functional integral over $\zeta,
\eta$ appears. Then, the integration variables $\varphi$ and $\psi$ are
introduced as a result of the action of the corresponding operator fields on
the coherent states. Therefore, the Frohlich Hamiltonian gives rise to the quantum
field theory of the interacting fermion and phonon fields,
with the partition function given by
%
\begin{equation}
Z = \int\!\!d \bar{\psi}\,d \psi\,d \varphi\, {\rm exp}\Bigl(i\!\!\int\!\!d^3x\,dt\,
\bar{\psi} [i \slashed{\p}-m-\charge\gamma^0 \varphi\,] \psi
+
i\!\!\int\!\!d^3x\,dt\,
\Big[
\frac{(\bm\nabla\dot{\varphi})^2}{2\Omega^2}
-
\frac{(\bm\nabla\varphi)^2}{2}
\Big]
\Bigr),
\label{ftfrohlich}
\end{equation}
%
where
$\slashed{\p}=\gamma\sp\mu\p\sb\mu$,
with the summation over $\mu=0,\,\dots,\,3$;
$
(\bm\nabla\varphi)^2
= \sum_{j=1}\sp{3}(\frac{\p \varphi}{\p x\sb j})^2.
$
Above,
we denote $\charge = \tilde{\alpha}\sqrt{\frac{1}{2\Omega}}$.
We formalize this as follows:
\begin{lemma}
The quantum-mechanical system with the Frohlich Hamiltonian (\ref{frohlich})
is equivalent to the field theory with the partition function (\ref{ftfrohlich}).
\end{lemma}

It is worth mentioning that usually the nontrivial Vierbein appears when
Dirac fermions emerge in condensed matter models \cite{MR1989766}.
Moreover, this Vierbein fluctuates.
Under certain circumstances, such fluctuations
may give rise to the emergent gravity \cite{MR1989766}.
We, therefore, assume that the
emergent Vierbein has small fluctuations and can be transferred to the unity matrix
via the rescaling of space and time coordinates. In this model the emergent
Lorentz symmetry is broken due to the Coulomb forces. That is why we deal with an
exotic situation: the fermion Hamiltonian is of the Dirac form, while the
interactions are purely nonrelativistic.

In the low energy approximation $E\ll\Omega$,
we arrive at the partition function
with the action given by Eq. (\ref{gact}):

\begin{equation}
Z  =  \int\!\!d \bar{\psi}\,d \psi\,d \varphi\, {\rm exp}\Bigl(i\!\!\int\!\!d^3x\,dt\,
\bar{\psi} [i \slashed{\p}-m -\charge\gamma^0 \varphi\,] \psi
-
i\!\!\int\!\!d^3x\,dt\,
\frac{(\bm\nabla\varphi)^2}{2}\Bigr). \label{Z_p}
\end{equation}
After the Wick rotation ($t\rightarrow -it$, $\varphi \rightarrow i \varphi $)
to Euclidean space-time we arrive at
\begin{equation}
Z = \int\!\!d \psi\, d \bar{\psi}\,d \varphi\, {\rm exp} \Big(-\int\!\!d^3x\,dt\,
\bar\psi
[\Gamma^0(\p_0 + i \charge\varphi)+\Gamma\sp j\p\sb j+m]
\psi
+\int\!\!d^3x\,dt\,
\frac{(\bm\nabla\varphi)^2}{2} \Big). \label{Z_E}
\end{equation}

Here the Euclidean gamma-matrices $\Gamma^{\mu}$
satisfy $\{\Gamma^{\mu}, \Gamma^{\nu}\}=2\delta^{\mu \nu} $,
$\ 0\le \mu,\nu\le 3$.

\begin{remark}\label{remark-above}
In fact, the integration in (\ref{Z_E})
is not convergent due to the positive sign
at the kinetic term for $\varphi$.
This is exactly the same problem as for the
Euclidean functional integral for the gravitational theory
with the Einstein--Hilbert action (see below).
This shows that the theory defined by the partition
function (\ref{Z_p}) can be considered only as an effective low energy model.
At some (large) energies, the action has to be redefined in order to make the
Euclidean functional integral convergent. As well as for the quantum gravity,
this can be done if the term with higher powers of $\partial \varphi$ is added
to the action. This, in turn, regularizes the Coulomb interaction at small
distances. In physical applications,
such additional term in
the action appears at scales
at which the attraction due to optical
phonons no longer dominates,
and some other interactions come into play
(say, the Coulomb repulsion due to photons).
\end{remark}

\subsection{Application of semiclassical methods to the model}

In this subsection we follow the approach of \cite{PhysRevD.12.2443}
on the semiclassical methods for fermion systems and obtain similar results.
Here it is important that we consider the
phonon field constant in time. We come to the model
with the following partition function:
$$
Z=
\int\!\!d\bar{\psi}\,d\psi\, d f\,
{\rm exp}
\Big(
i\sum\limits\sb{\eta}T\int\!\!d^3 x\,
{\psi}_{\eta}^+ [\eta-{\cal H}\sb{\varphi}] \psi_{\eta}
-
i\!\!\int\!\!d^3 x dt\, \frac{(\bm\nabla\varphi)^2}{2}
\Big),
$$
where
$$
{\cal H}\sb{\varphi}
=\gamma^0
\Big[
-i\sum\sb{j=1}\sp{3}\gamma\sp j \p\sb j+m +\charge\gamma^0 \varphi
\Big].
$$
Here the system is considered
with the anti-periodic in time boundary conditions:
$\psi(t+T, x) =-\psi(t, x)$.
We use the decomposition
\begin{equation}\label{dec-0}
\psi(t,x) = \sum_{\eta=\frac{\pi}{T}(2k+1),\, k \in\Z} e^{-i \eta t}
\psi_{\eta}(x).
\end{equation}
We represent $\psi$ as $\psi_{\eta}(x) = \sum_n
c_{\eta, n} \Psi^{\varphi}_n(x)$, where $\Psi^{\varphi}_n$ is the
eigenfunction of ${\cal H}\sb{\varphi}$
corresponding to the eigenvalue $E^{\varphi}_n$
and normalized to unity ($\int\!\!d^3x\,{\Psi}_n^+ \Psi_n = 1$):
$$
Z = \int\!\!d \bar{c} \,d c\,d \varphi\,
{\rm exp}
\Bigl(
i \sum\limits_{\eta,n} T
\bar{c}_{\eta, n} [\eta-E^{\varphi}_n] c_{\eta, n}
-i\!\!\int\!\!d^3x\,
\,dt\, \frac{(\bm\nabla\varphi)^2}{2}
\Bigr)
.
$$
Integrating out the Grassmann variables $c_n$ we come to:
\begin{eqnarray}
\label{G_p_h_E}
&&
Z=\int\!\!d \varphi\,
{\rm exp}
\Bigl(-i\!\!\int\!\!d^3 x\,dt
\frac{(\bm\nabla\varphi)^2}{2}
\Bigr)
\prod_{\eta}
\prod_{n}
\big((\eta-E^{\varphi}_n)T\big)
\nonumber
\\
&&
\ =C\int\!\!d \varphi\,
{\rm exp}
\Bigl(-i\!\!\int\!\!d^3x\,dt\,\frac{(\bm\nabla\varphi)^2}{2}
\Bigr)
\prod\sb{n}\cos\frac{ T E^{\varphi}_n}{2},
\end{eqnarray}
where
$C$ depends on the details of the regularization but does not
depend neither on $T$ nor on the spectrum
in the continuum limit. The values
$E^{\varphi}_n$ depend on the parameters of the Hamiltonian,
with the index $n$ enumerating these values.

Eq. (\ref{G_p_h_E}) is derived as follows.
Recall that
in \eqref{dec-0}
the summation is over
$\eta=\frac{\pi}{T}(2k+1)$.
The product over $k$ can be calculated as in \cite{PhysRevD.12.2443}:
\begin{equation}
\prod_{k\in\Z}
\Big(1 + \frac{ E^{\varphi}_n T}{\pi(2k+1)}\Big)
=
\cos\frac{E^{\varphi}_n T}{2},
\label{DET0}
\end{equation}
where $T = 2 N a$.
Here we imply that the lattice regularization is introduced,
and $a$ is the lattice spacing while $2N$ is the lattice size in the imaginary time direction.
In the limit $a\rightarrow 0$ we come to $N \rightarrow \infty$.
Thus,
\begin{equation}
 {\rm Det} ( i \partial_{0}-{\cal H}_{\varphi})
=C\prod_n \cos\frac{E^{\varphi}_n T}{2}.
\end{equation}
We get (see also \cite{PhysRevD.12.2443,Rajaraman1975227}):

\begin{eqnarray}
&&
\hskip -40pt
Z
=
C\sum_{\{K_n\} = 0,1} \int\!\!d \varphi\,
{\rm exp}
\Bigl(
-
i T\int\!\!d^3x\,
\frac{(\bm\nabla\varphi)^2}{2}
+ \frac{i T}{2} \sum_n
E_n^{\varphi}-i T \sum_n K_n E_n^{\varphi}
\Bigr)
\nonumber\\
&&
\hskip -40pt
\quad
=
C'\sum_{\{K_n\} = 0,1} \, \int\!\!d \varphi\,
e^{i Q_{K_n}(\varphi)}.
 \label{G_p_h_E_}
\end{eqnarray}

Following \cite{PhysRevD.12.2443},
we interpret Eq. (\ref{G_p_h_E_}) as follows. $K_n$
represents the number of occupied states with the energy $E_n^{\varphi}$. These
numbers may be $0$ or $1$. The term $\sum_n E_n^{\varphi}$ vanishes
if $\varphi = 0$,
since in this case the values $E_n$ come in pairs with the opposite signs.

In the weak coupling approximation when $\frac{\charge^2}{4\pi}\ll 1$
the energy levels can be represented as
$E_n^{\varphi} \approx E_n^{0} + E_n^{\prime}\varphi$.
Then the integral over $\varphi$ is Gaussian;
it is equal to $\sim {\rm exp}(i Q_{K_n}(\varphi_{class}) T)$,
where $\varphi_{class}$ satisfies the variational problem
$\delta Q_{K_n}(\varphi)=0$.
%
%
In this limit,
the dominant contributions of the $\varphi$-configurations
satisfy the following variational problem:

\begin{eqnarray}
0 & = & \delta
\left[ -
\int\!\!d^3x\,
\frac{(\bm\nabla\varphi)^2}{2}
+
\sum_n
\Big(
\frac{E_n^{\varphi}}{2}-K_n E_n^{\varphi}
\Big)
\right]
\nonumber\\
&=&\delta
\int\!\!d^3x
\left[
-
\frac{(\bm\nabla\varphi)^2}{2}
-\sum_n
\Big(K_n
\zeta_n^+ {\cal H}\sb{\varphi} \zeta_n
-
\frac{\zeta_n^+ {\cal H}\sb{\varphi} \zeta_n}{2}
\Big)
\right].
\end{eqnarray}

In the right-hand side, the variation over $\zeta_n$ with the constraint $\int\!\!d^3x\,
\zeta_n^+ \zeta_n = 1$ gives, in addition, the one-fermion wave functions
$\zeta_n$. The variational problem can be written as

\begin{eqnarray}
0&=&\delta\int\!\!d^3 x
\left[
\sum_n \Big(\frac 1 2-K_n\Big)\zeta_n^
+[{\cal H}\sb{\varphi}-\lambda_n]\zeta_n
-\frac{(\bm\nabla\varphi)^2}{2}
\right]
\nonumber\\
&=&
\delta\int\!\!d^3 x
\left[
\sum_n \Big(K_n-\frac 1 2\Big)\bar{\zeta}_n
[i \slashed{\p}-m-\charge\gamma^0 \varphi]{\zeta}_n
-\frac{(\bm\nabla\varphi)^2}{2}
\right].
\label{var}
\end{eqnarray}

Here $\lambda_n$ are the Lagrange multipliers.
We introduced the
time dependence into $\zeta$: the variation is performed with respect to the
functions of the form $\zeta_n = e^{-i\lambda_n t} \zeta(x)$ and with
respect to
the time-independent
phonon cloud $\varphi$.
In this form,
the functional to be used in the variational problem almost coincides with the
action from (\ref{ftfrohlich}).
The difference is that we assume
the special form of
$\zeta \sim e^{-i \lambda t}$
and also that $\varphi$ does not depend on time.
Also, instead of the Grassmann variables,
we substitute ordinary wave functions and take into
account filling factors for the fermion states.

The additional constraint is that the wave functions $\zeta_n$ are different,
so that there are no states that are occupied more than once.
The variation is performed
with the numbers $K_n$ being fixed.
After the variational problem is solved one says that the state
with the wave function $\zeta_n$ that has been found
is occupied if $K_n\ne 0$ for the corresponding value of $n$.
In Eq. (\ref{G_p_h_E_}) we need to sum up all such configurations with different arrays $K_n$.
The calculated values of $\varphi$
are to be substituted into the exponent in Eq. (\ref{G_p_h_E_})
while the values of $E_n$ are given by
$E_n = {\zeta}^+_{n} {\cal H} \zeta_{n}$.
This approach is similar to the conventional Hartree--Fock approximation.

We come to the following result:

\begin{theorem}
The partition function for the system of Dirac electrons interacting with
optical phonons at low energies $E \ll\Omega$ is given by Eq. (\ref{G_p_h_E_}).
In the weak coupling limit,
the integral over $\varphi$ in (\ref{G_p_h_E_}) is evaluated
in the stationary phase approximation,
resulting in the variational problem (\ref{var}).
\end{theorem}

\begin{remark}
CP-invariance implies that for any state $n$ with the energy $E_n^{\varphi}$
there exists a state $\tilde{n}$ with the energy
$E_{\tilde{n}}^{\varphi} = - E_n^{-\varphi}$.
In particular, for $\varphi = 0$ we obtain $\sum_n E_n = 0$.
States with positive values $E_n$
are interpreted as electrons, while the states with
negative $E_n$ correspond to holes. The vacuum state is in this case the state
with all negative levels $E_n$ occupied. The situation is changed when
$\varphi \ne 0$. However, if ${\rm max}\, \charge\varphi \le m$,
the states with $E_n> 0$
are also interpreted as electrons while the states with $E_n < 0$ are
interpreted as holes. If ${\rm max}\, \charge\varphi > 2m$, then
there could be states
which can not be considered as either electrons or holes.
Instead, these states
correspond to the Schwinger pair creation process. The appearance of such
states, however, may be avoided in the weak coupling limit,
when $\charge$ is small and,
therefore, one almost always has
$\charge\varphi<m$. That is why in the weak coupling
the vacuum can again be considered
as the state with all negative levels of $E_n$ occupied
and all positive levels of $E_n$ empty.
At large enough values of $\alpha =\frac{\charge^2}{4\pi}$
this pattern may be changed
due to the creation of pairs that may lead to the change of vacuum. In this case the
fermion condensate may appear;
the description of the theory in terms of the
collection of one-fermion states is no longer relevant.

\end{remark}

\subsection{One-polaron problem}

Now let us consider the usual polaron problem, i.e. the problem of one electron
interacting with the phonon cloud. Recall that $Z = {\rm Tr}\, e^{-i {\cal H}
t}$.
Therefore, the sum in Eq. (\ref{G_p_h_E_}) corresponds to the sum over
many-fermion states. The state with $K^{vac}_n = \frac{1}{2}(1-{\rm sign}\,
E_n)$ corresponds to the vacuum. Then the state with $K_n(q) = K^{vac}_n +
\delta_{n q}$ for some $E_q > 0$ corresponds to the state
that consists of the vacuum (the Dirac sea of negative levels)
and the bound state of one electron and
the phonon cloud surrounding it. This is a polaron.
The vacuum is translation-invariant and CP-invariant.
That is why $\varphi_{vac} = 0$. The vacuum energy is
defined as $E_{vac} = \sum_n K_n^{vac} E_n^{0}$. The polaron energy is equal to
\begin{equation}
E_q = \int\ d^3x\, \frac{(\bm\nabla\varphi_q)^2}{2}
+\sum_n \Big(K_n(q)-\frac 1 2\Big) E_n^{\varphi_q},
\end{equation}
where $\varphi_q$ is defined by solving the variational problem (\ref{var}).
Infinite vacuum energy has to be subtracted from $E_q$;
the quantity ${\cal E}_q = E_q-E_{vac}$ is finite and is considered as a
renormalized polaron state energy in physical applications.
In the general case, the
renormalized polaron energy
contains the contribution from the virtual electron-hole pairs
which are born when large enough $\varphi_q$ appears.
The interaction with the Dirac sea also contributes into ${\cal E}_q$.
However, under certain circumstances, these contributions can be neglected.
This is the so-called quenched approximation when the interactions
between different fermions are neglected while the interaction between the phonon
cloud and the electron is taken into account. This means that the probability
that the electron-hole pair is created from vacuum is small,
while the total electric charge $Q$ for the given problem is implied equal to unity.
For the case of the conventional polaron these conditions are assumed
\cite{froelich-1954, polaron_lectures, 1990AnPhy.202..320C, polaron-1933, polaron-1946}.
In our case the setup for the one-polaron problem should include
$Q=1$ and $|{\cal E}_q - m| \ll m$.
The latter condition provides that the probability for the electron-hole pairs
to be created is small.
In this case, the polaron energy is calculated as
\begin{equation}
{\cal E}_q =
\int\!\!d^3x\, \frac{(\bm\nabla\varphi_q)^2}{2} +
E_q^{\varphi_q},
\end{equation}
where $\varphi_q$
is calculated via the variational problem
\begin{equation}
0 = \delta \int\!\!d^3x\,
\Big\{\bar{\zeta} [i \slashed{\p}-m -\charge\gamma^0 \varphi_q ]
\zeta
-
\frac{(\bm\nabla\varphi_q)^2}{2}\Big\}.
\label{var___}
\end{equation}
This problem, in turn, leads to the following equations:
\begin{eqnarray}
\label{2dd}
&&
i\p\sb{t}\zeta
=-i
\bm\alpha\!\cdot\!\bm\nabla
\zeta+m\beta\zeta
+\charge\varphi_q(x,t)\zeta(x,t),
\\
\label{maxwell-0}
&&
\Delta\varphi_q(x,t) =\charge\abs{\zeta(x,t)}^2,
\end{eqnarray}
where
$x\in\R^3,$
$\zeta(x,t)\in\C^{4}$.
That is,
the potential $\varphi_q(x,t)\in\R$
is generated by the spinor field itself.
That is why we come to Eq. (\ref{2dd-0}) with the important constraint on the wave function $\zeta$:
\begin{equation}
\int\!\!d^3x\,
\zeta^+ \zeta = 1.
\label{constrnt}
\end{equation}

We arrive at the following statement:

\begin{lemma}
\label{polaronenergy} Let us consider the problem of bound states of one
electron surrounded by the phonon cloud (polaron) in the field theory with
partition function (\ref{ftfrohlich}).
In the low energy approximation $E \ll \Omega$
the system of equations (\ref{2dd}), (\ref{maxwell-0}), (\ref{constrnt}),
as well as the
variational problem (\ref{var___}), solve this problem in the first
order
approximation of the weak coupling expansion.
\end{lemma}

\begin{remark}
It is important that only those solutions of the system (\ref{2dd-0})
which are normalized according to Eq. (\ref{constrnt})
have a physical meaning.
\end{remark}

It is worth mentioning that the given variational problem may be relevant for
the solution of polaron problem not only in the weak coupling regime (see, for
example, \cite{polaron_lectures,polaron-1933,polaron-1946}).
In the nonrelativistic case, this variational problem
has appeared in the approach due to Landau and Pekar
\cite{polaron-1933,polaron-1946}.
However, in these papers
the trial
functions were used for the minimization, while we consider the given
variational problem exactly.

It is instructive to consider how Eq. (\ref{var___}) appears from the
consideration of the two-point Green function

\begin{equation}
G(t_2-t_1) = \frac{1}{Z}\int\!\!d \bar{\psi}\,d \psi\,d \varphi\,
e^{
i\!\!\int\!\!d^3x\,dt\,
\big\{
\bar{\psi} [i \slashed{\p}-m -\charge\gamma^0 \varphi
\,] \psi -
\frac{(\bm\nabla\varphi)^2}{2}
\big\}
}
{\psi}^+(t_1,x) \psi(t_2,x)\,d^3x\,.
\label{GGG_p_h_}
\end{equation}

Let us mention that the consideration of arbitrary values of $t_1, t_2$
requires a more complicated technique, when $\varphi$ is constant as a function
of time except at the points $t_1, t_2$. This is because the insertion of
$\psi$ and $\bar{\psi}$ disturbs vacuum in such a way that the value of
$\varphi$ is changed at $t_1$ and $t_2$. This technique uses the so-called
Floquet indices and will be applied in the next section to the consideration of
gravitational polarons, which are the relativistic generalizations of the objects
considered in this section. Here we restrict ourselves
to the case $t_1 = 0, t_2 = T$. Then

\begin{equation}
G(T)=\frac{{\rm const}}{i T Z}\sum_{\{K_n\} = 0,1}
\int\!\!d\varphi\,
e^{-i T\int\!\!d^3x\, \frac{(\bm\nabla\varphi)^2}{2}}
e^{
i T\sum_n
\big(\frac{1}{2}-K_n \big)E_n^{\varphi}
}
\sum_{\eta = \frac{\pi}{T}(2k+1),q}\frac{-e^{-i\eta T}}{\eta-E_q^{\varphi}}.
\end{equation}

Using the Poisson summation formula \cite{PhysRevD.83.025005,PhysRevD.19.2385},
we get
$$
g(t,E_q^{\varphi}) =
\sum_{\eta = \frac{\pi}{T}(2k+1)}\frac{e^{-i\eta t}}{\eta-E_q^{\varphi}}
=\frac{-i T e^{-iE_q^{\varphi}t} }{1+e^{-iE_q^{\varphi}T}}.
$$
Here it is implied that the energy levels have small imaginary parts
(as usual in the quantum field theory).
After the Wick rotation ($T\rightarrow -i/{\cal T} $, $\cal T$
being the temperature)
the given expression would become the usual
finite temperature  Matsubara Green function.
Therefore, we arrive at

\begin{equation}
G(T)=\sum_q \frac{ \sum\limits\sb{\{K_n\} = 0,1} \int\!d\varphi\,
e^{
i T\big(-\int\!d^3x\, \frac{(\bm\nabla\varphi)^2}{2}
+
\sum_n \frac{E_n^{\varphi}}{2}
-\sum_{n\ne q} K_n E_n^{\varphi}-E_q^{\varphi}
\big)
}
}
{\sum\limits\sb{\{K_n\} = 0,1} \int\!d\varphi\,
e^{i T\big(
-
\int\!d^3x\, \frac{(\bm\nabla\varphi)^2}{2}
+\sum_n \frac{E_n^{\varphi}}{2}
-\sum_n K_n E_n^{\varphi} \big)}
}.
 \label{GGG_p_h_E_}
\end{equation}

Here in each term of the summation over $K_n$
the field $\varphi$ is to be determined
in the stationary phase approximation via solving the variational problem (\ref{var}).
Again, in the quenched approximation we come to
$$
G(T)=\sum_q \int\!\!d \varphi\,
e^{
-
i T
[\int\!\!d^3x\, \frac{(\bm\nabla\varphi)^2}{2}+E_q^{\varphi}]
}
= \sum_q Z_q e^{-i{\cal E}_q
T}.
$$
Here ${\cal E}_q$ is the energy of the polaron in the $q$th state calculated
according to Lemma~\ref{polaronenergy}.
The factors $Z_q$ are the pre-exponential
factors of the stationary phase approximation.


\section{Gravitational polarons}
\label{sectgrav}

\subsection{Semiclassical description of fermions
coupled to the gravitational field}

In the previous section,
we considered the quantum system of Dirac electrons
interacting via the attractive Coulomb potential. This system may appear in
certain nonrelativistic models at the quantum phase transition between the
phases of the fermionic systems with different values of topological invariants
\cite{MR2911335}. However, the Coulomb interaction breaks the emergent Lorentz
symmetry. It is interesting, therefore, to consider the relativistic extension
of the model defined by the partition function (\ref{Z_p});
one such extension is discussed in this section.

We consider the relativistic Dirac fermion interacting with the gravitational
field. The action of a Dirac spinor in Riemann space has the form
\cite{doi:10.1142/S0217732310034110,shapiro2002physical,PhysRevD.70.105004,PhysRevD.84.124042,2012arXiv1208.1254V}
\begin{equation}
S_f = \int
\big(
i\bar{\psi}{\bm \gamma}^{\mu} D\sb{\mu}\psi-m\bar{\psi}\psi
\big)
|E|\,d^4x\,.
\label{Sf}
\end{equation}
Here $|E| = {\rm det} E^a_{\mu}$, where $E^a_{\mu}$ is the inverse Vierbein,
${\bm \gamma}^{\mu} = E^{\mu}_a \gamma^a$,
and
$\bar{\psi} = \psi^+ \gamma^0$.
The covariant derivative is
\begin{equation}
D\sb{\mu} =
\partial_{\mu} + \frac{1}{4} \omega_{\mu}^{ab} \gamma_{[a}\gamma_{b]},
\label{GD}
\end{equation}
with
$\gamma_{[a}\gamma_{b]} =
\frac{1}{2}(\gamma_{a}\gamma_{b}-\gamma_{b}\gamma_{a})$.
 The torsion-free spin
connection is denoted by $\omega_{\mu} $.
It is related to $E^a_{\mu}$ and
the affine connection $\Gamma^{\rho}_{\mu\nu}$ as follows:
\begin{eqnarray}
\bm\nabla_{\nu} E_{\mu}^{a} &=& \partial_{\nu}E^a_{\mu}-\Gamma^{\rho}_{\mu
\nu}E^a_{\rho} + \omega^{a}_{ . b\nu}E^b_{\mu}=0,
\nonumber\\
{D}_{[\nu} E_{\mu]}^{a} &=& \partial_{[\nu}E^a_{\mu]} + \omega^{a}_{.b[\nu
}E^b_{\mu]}=0.
\end{eqnarray}
This results in:
\begin{eqnarray}
&&
\Gamma^{\rho}_{\mu \nu}= \{^{\rho}_{\mu \nu}\} =
\frac{1}{2}g^{\alpha\lambda}(\partial_{\beta}g_{\lambda
\gamma}+\partial_{\gamma}g_{\lambda \beta}-\partial_{\lambda}g_{\beta \gamma}),
\nonumber\\
&&\omega_{ab\mu} = \frac{1}{2}( c_{abc}-c_{cab}+c_{bca})E^{c}_{\mu}.
\end{eqnarray}
Here $c_{abc} = \eta_{ad}E^{\mu}_b E_c^{\nu}\partial_{[\nu}E^d_{\mu]}$,
$\ g_{\mu
\nu} = E^a_{\mu}E^b_{\nu}\eta_{ab}$,
and
$\ \Gamma^{\rho}_{\mu\nu}-\Gamma^{\rho}_{\nu \mu} = 0$;
indices are lowered and lifted with the aid of $g$ and $E$.

The partition function of the model is given by
$$
Z = \int\!\!d \bar{\psi}\,d \psi\,d E\,
e^{i\!\!\int\!\!d^4x\,
\big(
|E|
\bar{\psi} [i\slashed{D}-m] \psi-\frac{1}{16 \pi G} R |E|\big)
}
$$
where $\slashed{D}={\bm \gamma}\sp\mu D\sb\mu$,
with $D_{\mu}$ given by Eq. (\ref{GD}).

\begin{remark} The integral over the Grassmann variables $\psi$ in continuum field
theory requires additional discussion. There are several ways to
define the functional integral: via lattice discretization, via re-expressing it
as a functional determinant, etc.
In all these cases, the presence of a nontrivial metric
leads to additional difficulties.
Below we assume that the functional integral
is defined in such a way that
\begin{equation}
\int\!\!d \bar{\psi}\,d \psi \,
e^{i\!\!\int\!\!d^4x\, |E| {\psi}^+ \hat{Q}
\psi}
= {\rm Det}\, \hat{Q} = \prod_n \lambda_n, \label{DET_}
\end{equation}
where $\lambda_n$ are eigenvalues of $\hat{Q}$. (The spectrum of $\hat{Q}$ is
discrete if we consider the system in the finite four-volume $V_4 = \int\!\!d^4x\,
|E|$.)
We assume the toroidal topology,
and also that
in a synchronous reference frame
the boundary conditions
are antisymmetric in time and symmetric in the spatial coordinates.
Eq. (\ref{DET_}) can be rewritten as
\begin{equation}
{\rm Det}\, \hat{Q}=\int\!\!d \bar{\psi}\,d \psi \,
e^{
i \sum_n
\lambda_n \int\!\!d^4x\, |E| {\Psi}_n^+(x) \Psi_n(x) \bar{c}_n c_n
}
= \int\!\!d \bar{c} \,d c
\,{\rm Det}\frac{\partial(\bar{c},c)}{\partial(\bar{\psi},\psi)}
\,
e^{i \sum_n \lambda_n \bar{c}_n c_n
},\label{DET}
\end{equation}
where we used the decompositions
\begin{equation}
\psi = \sum_n \Psi_n(x) c_n,
\qquad
\psi^+ = \sum_n \Psi^+_n(x)
\bar{c}_n.
\label{dec}
\end{equation}
Here $\Psi_n(x)$ are eigenfunctions of $\hat{Q}$ corresponding to eigenvalues
$\lambda_n$, while $c_n, \bar{c}_n$ are new Grassmann variables.
The operator
$\hat{Q}$ is assumed to be hermitian with respect to the inner product $\langle
\psi, \phi \rangle = \int\!\!d^4x\, |E| \psi^+ \phi$. Therefore, the eigenfunctions
satisfy $\int\!\!d^4x\, |E| {\Psi}_n^+(x) \Psi_m(x) = \delta_{n m}$.
The key
assumption about the integration measure over $\psi$ is that with this
normalization one has
${\rm Det}\frac{\partial(\bar{c},c)}{\partial(\bar{\psi},\psi)}=1$.
Eq. (\ref{DET}) also allows us to calculate various correlation
functions $\langle \psi^+(x_1)\dots \psi(x_N)\rangle$.
Namely, we first represent $\psi$
as the series (\ref{dec}),
and then the integral over $\bar{c},c$ is evaluated as usually.
\end{remark}

Again, the integral in the Euclidean space is not convergent, and should be redefined
at high energies; see the discussion above (Cf. Remark~\ref{remark-above}).
In order to bring the theory into the
form suitable for the considerations similar to that of \cite{PhysRevD.12.2443},
let us
consider the system in the gauge corresponding to a synchronous reference
frame. In this gauge, $E^0_a E^\mu_b \eta^{ab} = \delta^{0\mu}$;
we also set
$E^0_a = \delta^0_a$ via the rotation of the
reference frame
in its internal space and the
corresponding $SO(3,1)$ transformation of spinors. That is why the gauge is
fixed both with respect to the general coordinate transformations and with respect
to the inner $SO(3,1)$ rotations of the reference frame. We denote
$$
\gamma^{0}
\big[i E^{\mu}_a \gamma^{a}
\big(\partial_{\mu}+ \frac{1}{4}
\omega_{\mu}^{ab} \gamma_{[a}\gamma_{b]}\big) -m
\big]
=
i\partial_{0}-{\cal H},
$$
$$
-{\cal H} = i E\sp j_b \gamma^{0} \gamma^{b}
\p\sb j+ i E^{\mu}_b \gamma^{0}\gamma^{b}\frac{1}{4} \omega_{\mu}^{c d}
\gamma_{[c}\gamma_{d]} -\gamma^{0} m ,
$$
where the summation in $j$ is over $j = 1,\,2,\,3$.

We point out that the operator $i
\partial_{0}-{\cal H}$ is hermitian (while $\cal H$ is not). We have
$$
Z  =  \int\!\!d E\, {\rm Det} ( i
\partial_{0}-{\cal H} )\,
e^{-\frac{i}{16 \pi G} \int\!\!d^4x\,R |E|}.
$$

In order to calculate the determinant $ {\rm Det} ( i
\partial_{0}-{\cal H}) $,
we use anti-periodic in time boundary conditions.
Suppose that we find the solution $\zeta$ of
the equation $( i \partial_{0}-{\cal H} )\zeta = 0$ such that
$\zeta_{\Omega}(t + T) = e^{-i \Omega T} \zeta_{\Omega}$. (Here $\Omega T$ is
the Floquet index \cite{PhysRevD.12.2443}). Then $\Psi_{k, \Omega} =
e^{i\frac{\pi}{T}(2k+1)t + i\Omega t}\zeta_{\Omega}$ is the eigenfunction of
the operator $( i
\partial_{0}-{\cal H} ) $:
\begin{equation}
( i \partial_{0}-{\cal H} )\Psi_{k,\Omega}
=-\Big(\frac{\pi}{T}(2k+1) +
\Omega\Big)\Psi_{k,\Omega}.
\end{equation}
With the derivation similar to that of Eq. (\ref{G_p_h_E})
and Eq. (\ref{DET0})
we come to
\begin{equation}
{\rm Det} ( i
\partial_{0}-{\cal H}) ={\rm const}\, \prod_n \cos\frac{\Omega_n^E T}{2}.
\end{equation}
Here $\rm const$ depends neither on the gravitational field
nor on $T$,
while the product is over the different values $\Omega_n^E  T$
of the Floquet index \cite{PhysRevD.12.2443}.
These values depend on the Vierbein field $E^{\mu}_A$. The index $n$ enumerates them.
The partition function takes the following form
(Cf. Eq. (\ref{G_p_h_E_})):

\begin{equation}
Z \sim {\rm const}\,\sum_{\{K_n\} = 0,1} \int\!\!d E\, {\rm exp}\Bigl(- i
m_P^2 \int\!\!d^4x\,|E| R
+ i \frac{T}{2} \sum_n \Omega_n^E -i T \sum_n K_n \Omega_n^E \Bigr).
 \label{G_p_h_E_2}
\end{equation}

Here $m_P$ is the Planck mass.
Following \cite{PhysRevD.12.2443}, we interpret Eq.
(\ref{G_p_h_E_2}) as follows.
The numbers $K_n$ represent the number of occupied
states with the Floquet index $\Omega_n^E T$. These numbers may be $0, 1$.
The vacuum here corresponds to the negative ``energies'' (Floquet indices)
occupied and positive ``energies'' empty.

Here the physical meaning of the numbers $K_n$ is the same as in the previous
section.
The only difference is that now the Floquet indices appear in place of the
energy levels and that the gravitational field depends on time.
The semiclassical
approximation for the gravitational field now leads to the variational problem

\begin{eqnarray}
0 &=& \delta \Big\{-\sum_n [K_n-\frac12] \Omega_n^E T-m_P^2 \int\!\!d^4x\, R
|E|\Big\} \\
&=& \delta \int\!\!d^4x\,
\Big\{\sum_n [K_n-\frac12]{\Psi}^+_{k, \Omega^E_n} ( i
\partial_{0}-{\cal H}+\frac{\pi}{T}(2k+1) )\Psi_{k, \Omega^E_n}-m_P^2
R \Big\}|E|.\nonumber
\label{var2_}
\end{eqnarray}

Here $E$ is varied, $k$ is arbitrary,
and the normalization is $\int\!\!d^4x\, |E|
{\Psi}^+_{k, \Omega^E_n} \Psi_{k, \Omega^E_n} = T$. Let us also introduce the
wave function
\begin{equation}
\zeta_{n,\lambda,E}=e^{-i(\frac{\pi}{T}(2k+1) + \lambda_n) t} \Psi_{k, \Omega^E_n}.
\end{equation}
The values $\lambda_n$ play the role of Lagrange multipliers. At $\Omega_n =
\lambda_n$, the functions $\zeta_{n,\lambda,E}$ satisfy the following conditions:
$$
\zeta_{n,\lambda,E} (t+T) = e^{-i\lambda T} \zeta_{n,\lambda,E}(t),
$$
$$
\int\!\!d^4x\, |E|
\bar{\zeta}_{n,\lambda,E} E^0_A\gamma^A \zeta_{n,\lambda,E} = T,
$$
\begin{equation}
[i \slashed{D}-m ]\zeta_{n,\lambda,E} = 0.\label{DIRAC}
\end{equation}


Eq. (\ref{var2_}) can be rewritten as
\begin{equation}
0 =\int\!\!d^4x\,\Bigl(\sum_n [K_n-1/2]\bar{\zeta}_{n, \lambda, E} \{\delta[i
\slashed{D}
-m ]|E|\} \zeta_{n, \lambda, E}-m_P^2 \{\delta R |E|\} \Bigr).
\label{var2_2}
\end{equation}

One can see that the variation of $\zeta_{n, \lambda, E}$
does not enter this expression that defines the field $E$
for any given $\zeta_{n, \lambda, E}$.
At the same time one can see that the variation of $\zeta_{n, \lambda, E}$
would give the Dirac equation (\ref{DIRAC}). That is why
for the determination of both $\zeta_{n, \lambda, E}$ and $E$ we may use
the variational problem
\begin{equation}
0=
\delta \int\!\!d^4x\,\big\{\sum_n
\big(K_n-\frac12\big)\bar{\zeta}_{n, \lambda, E}
[i \slashed{D}-m ]
\zeta_{n, \lambda, E}
-m_P^2 R \big\}|E|,
\label{var2}
\end{equation}
where
the gravitational field and the wave functions $\zeta_n$ are varied. The
wave functions are normalized so that $\int\!\!d^4x\, |E| {\zeta}^+ \zeta = T$.
The additional constraint is that the wave functions $\zeta_n$ are different,
so that there are no states that are occupied more than once.
The variation is performed with the fixed values of $K_n$.
The final form of the variational problem is gauge invariant,
although it was
derived in a synchronous reference frame.
As a result, the gravitational field is
defined by the Einstein equations with the energy-momentum tensor defined by
the set of one-fermion states that represent the sea of occupied energy levels
and the fermion-antifermion excitations given by the set $K_n$. Fermion wave
functions are defined by the Dirac equation in the given external gravitational
field.

After the variational problem is solved one says that the state
with the wave function $\zeta_n$ that has been found
is occupied if $K_n\ne 0$ for the corresponding value of $n$.
In Eq. (\ref{G_p_h_E_2}) we need to sum up all such configurations
with different arrays $K_n$.
The calculated values of $E$ are to be substituted
into the exponent in Eq. (\ref{G_p_h_E_2})
while the values of $\Omega^E_n$ are given by
$\Omega^E_n = {\zeta}^+_{n} {\cal H} \zeta_{n}$.
This is the generalization of the Hartree--Fock approximation.

We came to the following result:

\begin{theorem}
The partition function for the system of Dirac fermions interacting with the
gravitational field is given by Eq. (\ref{G_p_h_E_2}).
In the semiclassical approximation
(for the energies $E \ll m_P$) the stationary phase approximation
leads to the variational problem (\ref{var2}).
\end{theorem}

\begin{remark}
\label{remarknorm}
The variational problem (\ref{var2}) is gauge invariant,
while the normalization of the wave function is not. We need
\begin{equation}
\int\!\!d^4x\,|E| {\zeta}^+ \zeta = T,
\end{equation}
where the spinor $\zeta$ and the time extent $T$
are defined in a synchronous reference frame.
Here $T$ is a global characteristic of the space-time, $|E|d^4x\,$ is the
invariant $4 $-volume, while $\zeta^+\zeta = \bar{\zeta}E^0_a\gamma^a\zeta$ is
the time-component of the $4$-vector.
\end{remark}

\subsection{Gravitational one-polaron problem}

In order to investigate one-polaron states, we consider the two-point Green
function

\begin{equation}
G(t_2-t_1) = \frac{1}{Z}\int\!\!d \bar{\psi}\,d \psi\,d E\,
e^{
i\!\!\int\!\!d^4x\, |E| \bar{\psi} [i \slashed{D} -m] \psi
-\frac{i}{16 \pi G}
\int\!\!d^4x\,R |E|}
{\psi}^+(t_1,x) \psi(t_2,x)\,d^3x\,|E(t_2,x)|.
\label{GG_p_h_}
\end{equation}

Here again the system is considered in a synchronous reference frame.
It is implied that the corresponding terms are present
in the integration measure over $E$.
With all the notations introduced above, we can represent $G$ as follows:

$$
G(t_2-t_1)  =  \frac{1}{\tilde{Z}} \sum_q \sum_{\{K_n\} = 0,1} \int\!\!d
E\,
e^{
-i m_P^2 \int\!\!d^4x\,|E| R-i T
\sum_{n\ne q} [K_n-1/2] \Omega_n^{E}-i (t_2-t_1) \Omega_q^{E}
}
F^E(t_1,t_2),
$$

where
$$
F^E(t_1,t_2) =-\int [{\psi}^E_q(t_1,x)]^+ \psi^E_q(t_2,x)
|E(t_2, x)|\,d^3x,
$$
$$
\tilde{Z} = \sum_{\{K_n\} = 0,1}
\int\!\!d E\,
e^{-i m_P^2 \int\!\!d^4x\,|E| R-i T \sum_n (K_n-\frac 1 2)
\Omega_n^{E}}
.
$$
Here large values of $m_P$ allow us to calculate integrals over $E$
in the stationary phase approximation.
When varying the terms in the exponent, it is necessary to take
into account $F^E(t_1,t_2)$ which disturbs the effective action at
$t = t_1,t_2$. However, if $t_2-t_1 = T$, then, as in the previous section,
we come to the following simplification:

\begin{equation}
G(T)  =  \frac{1}{\tilde{Z}} \sum_q \sum_{\{K_n\} = 0,1} \int\!\!d E
\, {\rm exp}\Bigl(-i m_P^2 \int\!\!d^4x\,|E| R-i T \sum_{n\ne q} [K_n-1/2]
\Omega_n^{E}-i T \Omega_q^{E} \Bigr).
\label{GG_p_h_E_2}
\end{equation}

The lemma follows:

\begin{lemma}
In the quenched approximation, at the energies much less than $m_P$,
the one-polaron problem is reduced to the variational problem
\begin{equation}
0=\delta \int\!\!d^4x\,
\Big\{\bar{\zeta} [i \slashed{D}-m ] \zeta-m_P^2 R
\Big\}|E|.
\label{var21}
\end{equation}
Here $\zeta$ is the fermion wave function normalized according to
Remark~\ref{remarknorm}.
\end{lemma}

\subsection{Newtonian limit}

In the nonrelativistic limit,
the energy-momentum tensor for the Dirac
field
is given by
$
T^{\mu\nu} = \bar{\zeta} i {\bm \gamma}^{\{\mu}
\partial^{\nu\}} \zeta\sim m {\zeta}^+\zeta\delta^{\mu 0}\delta^{\nu 0}.
$
The gravitational field is considered in the linear approximation
$g_{\mu\nu} =
\eta_{\mu\nu} + f_{\mu\nu}$,
where $\eta_{\mu\nu} = {\rm diag}\,(1,-1,-1,-1)$.
In the gauge $\partial^\mu h_{\mu\nu} = 0$
(where $h_{\mu\nu} = f_{\mu\nu} -
\frac{1}{2}\eta_{\mu\nu} f^\rho_\rho$)
Eq. (\ref{var21}) takes the form

\begin{equation}
0 = \delta \int\!\!d^4x\,
\left\{ \bar{\zeta}
\Big[i\slashed{\partial}
-m-\frac{i}{2}f^{\mu\nu}
{\gamma}_{\{\mu}\p_{\nu\}}\Big]
\zeta -\frac{m_P^2}{2}
\Big[ (\partial_\mu
f_{\nu\rho})^2-\frac{1}{2}(\partial_\mu f^\nu_{\nu})^2\Big]
\right\}.
\label{var3}
\end{equation}

Here the variation is performed with respect to the wave functions $\zeta$ and
with respect to the graviton cloud $f$.
As a result, the graviton cloud is formed in accordance with the
(linearized) Einstein equations with $T^{\mu\nu} = \bar{\zeta} i \gamma^{\{\mu}
\partial^{\nu\}} \zeta$.

In the nonrelativistic limit (we neglect gravitational waves that are not caused
by the given spinor field, see \cite[\S 99]{MR0475345}
on the Newtonian limit of general relativity and the definition of  $\phi$):
$$
f^0_0 = 2 \phi, \quad f^a_b = -2 \phi\, \delta^a_b,\quad a,\,b = 1,\,2,\,3;
$$
$$
i f^{\mu\nu}{\gamma}_{\{\mu}\p_{\nu\}} \zeta \sim 2 \phi m \gamma^0 \zeta.
$$
and $f^{0}_{a} = 0, f_0^a = 0$ for $a=1,2,3$.
We arrive at the following system of equations:
\begin{eqnarray}
&& \square \phi+\frac{m}{4 m_P^2} \bar{\zeta}\gamma^0 \zeta=0,
\nonumber \\
&& [i \slashed{\p}-m-m \phi \gamma^0 ] \zeta=0.
\end{eqnarray}

Next, we require that $\phi$ does not depend on time and denote
$ m \phi = e \varphi$,
$\charge = \frac{m}{2m_P}$.
As a result we arrive at
\begin{equation}\label{2dd-gr}
\left\{
\begin{array}{l}
i\p\sb{t}\zeta
=-i
\bm\alpha\!\cdot\!\bm\nabla
\zeta+m\beta\zeta
+\charge\varphi(x,t)\zeta(x,t)
\\
\Delta \varphi(x,t)=\charge\zeta\sp\ast(x,t)\zeta(x,t)
\end{array}
\right.
\end{equation}
with the normalization
\begin{equation}
\int\!\!d^4x\, \bar{\zeta}^{\prime} \gamma^t \zeta^{\prime} \sqrt{-g} = T.
\end{equation}
Here $\zeta^{\prime}$ is spinor field in a synchronous reference frame,
$\gamma^t
= E^0_a\gamma^a$ is the time component of covariant gamma-matrices
(also in a synchronous reference frame).

The normalization of the spinor field is the subject of careful
investigation. In the reference frame defined by the harmonic gauge,
both the spinor field and the gravitational field $\phi$ are independent of time.
However, in a general situation,
this is not the case in a synchronous reference frame.
 We may
represent $\bar{\zeta}^{\prime} \gamma^t \zeta^{\prime} = \bar{\zeta}
\gamma^{\mu} \zeta \frac{\partial [x^{\prime}]^0}{\partial x^{\mu}} =
J^{\mu}\frac{\partial [x^{\prime}]^0}{\partial x^{\mu}}$, where the current
$J^{\mu}= \bar{\zeta} \gamma^{\mu} \zeta$
is defined in the original reference frame.
In this frame, in the weak coupling, we have $\sqrt{-g}\sim
1-2\phi$ and $J^{\mu}\frac{\partial [x^{\prime}]^0}{\partial x^{\mu}}\sim
\zeta^+ \zeta $.
(Recall that
$g^{\mu\nu} \approx
\eta^{\mu\nu} - f^{\mu\nu}$, $g_{\mu\nu} \approx
\eta_{\mu\nu} + f_{\mu\nu}$, indices for $f^{\mu\nu}$ are lowered and lifted by
$\eta_{\mu\nu}$ and $\eta^{\mu\nu} $.
Then $f^{00} = 2\phi, f^{ab} = +2 \phi \delta^{ab}$.)
The latter follows from the Hamilton-Jacobi
equation that defines a synchronous reference frame
$g^{\mu\nu}\frac{\partial
[x^{\prime}]^0}{\partial x^{\mu}}\frac{\partial [x^{\prime}]^0}{\partial
x^{\nu}}=1$.
In the Newtonian approximation $\phi \ll 1$,
and  $\frac{\partial [x^{\prime}]^0}{\partial x^{\nu}}, \nu = 1,2,3$
are of the same order as $\phi$.
Therefore, up to the terms linear in $\phi$,
we get $g^{00}[\partial_0 [x^{\prime}]^0]^2=[E^{0}_0 \partial_0 [x^{\prime}]^0]^2=1$.
Also $J^{\mu}, \mu = 1,2,3$ are of the same order as $\phi$.
Therefore, up to the terms linear in $\phi$,
we have $J^{\mu}\frac{\partial [x^{\prime}]^0}{\partial x^{\mu}}
= J^{0} \frac{\partial [x^{\prime}]^0}{\partial x^{0}}
= \bar{\zeta} \gamma^0 \zeta E^0_0 \partial_0 [x^{\prime}]^0 \sim
\zeta^+ \zeta $.

That is why we come to the following normalization (valid in
original reference frame in the weak coupling):
\begin{equation}
\int\!\!d^3x\,\zeta^+ \zeta \sqrt{-g}
\approx \int\!\!d^3x\,\zeta^+ \zeta
\Big(1-\frac{1}{m_P}\varphi\Big) = 1.
\label{constrntgrv}
\end{equation}

\begin{lemma}
In the Newtonian limit in the harmonic gauge,
the gravitational polaron problem is
reduced to the system of equations (\ref{2dd}), (\ref{maxwell-0}), (\ref{constrntgrv})
with $\charge =
\frac{m}{2m_P}$.
\end{lemma}

\section{Stability of solitary waves in the Dirac--Coulomb system}
\label{sect-stability}

\subsection{Existence of solitary waves}

In this section,
we substitute
$\zeta$ by $\frac{1}{\charge}\zeta$
and
$\varphi$ by $\frac{1}{\charge}\varphi$,
so that $\charge$ disappears from the system
\eqref{2dd-0}.
As a result,
instead of the normalization condition (\ref{constrnt})
(polarons in condensed matter systems) we have
\begin{equation}
\int\!\!d^3x\,
\zeta^+ \zeta = \charge^2.
\label{constrnt2}
\end{equation}
For gravitational polarons, we have the constraint
\begin{equation}
\int\!\!d^3x\,
\zeta^+ \zeta
\Big(1-\frac{1}{e m_P} \varphi\Big) = \charge^2.
\label{constrnt2gr}
\end{equation}

We express $\varphi=\Delta^{-1}\abs{\zeta}^2$,
where
$\Delta^{-1}$ in $\R^3$
is the operator of convolution with $-\frac{1}{4\pi \abs{x}}$,
and write the Dirac--Coulomb system \eqref{2dd-0}
as the following Dirac--Choquard equation:
\begin{equation}\label{2dd-1}
i\p\sb{t}\zeta
=-i
\bm\alpha\!\cdot\!\bm\nabla
\zeta+m\beta\zeta
+\zeta\Delta^{-1}\abs{\zeta}^2,
\end{equation}
where
$\zeta(x,t)\in\C^{4}$,
$x\in\R^3$.
The solitary wave solutions
$\phi\sb\omega e^{-i\omega t}$
with $\omega\lesssim m$
can be constructed by rescaling
from the solutions to the nonrelativistic
limit of the model.
Such a method was employed
in
\cite{MR1750047,2008arXiv0812.2273G}
for the nonlinear Dirac equation
and in
\cite{MR2593110,MR2671162,MR2647868}
for the Einstein--Dirac system
and the Einstein--Dirac--Maxwell system;
for the Dirac--Maxwell system,
this approach has been implemented in
\cite{dm-existence}.
Let us mention that the solitary wave
solutions to \eqref{2dd-1}
with $\omega\lesssim m$
correspond to the
solitary wave solutions
of the Dirac--Maxwell system
with $\omega\gtrsim -m$
when the magnetic field is neglected;
such solitary waves were numerically obtained
in \cite{MR1364144}.
The sign change of $\omega$
is due to the different
sign of the self-interaction:
in the Dirac--Maxwell system,
the self-interaction is repulsive
for $\omega\lesssim m$
and attractive
for $\omega\gtrsim -m$;
in the Dirac--Choquard equation \eqref{2dd-1}, it is the opposite.
The profile
of the solitary wave
$\zeta(x,t)=\phi\sb\omega(x)e^{-i\omega t}$
satisfies
\begin{equation}\label{omega-phi-is}
\omega\phi\sb\omega
=-i
\bm\alpha\!\cdot\!\bm\nabla
\phi\sb\omega+m\beta\phi\sb\omega+
\phi\sb\omega
\Delta^{-1}\abs{\phi\sb\omega}^2.
\end{equation}
Let
$
\phi\sb\omega(x)=\begin{bmatrix}\phi\sb e(x,\omega)\\\phi\sb p(x,\omega)\end{bmatrix},
$
with
$
\phi\sb e,\ \phi\sb p\in\C^2$
the ``electron'' and ``positron'' components.
In terms of
$\phi\sb e$ and $\phi\sb p$,
\eqref{omega-phi-is} is written as
\begin{eqnarray}
\label{omega-phi-is-2}
\omega\phi\sb{e}
=-i\bm\sigma\!\cdot\!\bm\nabla\phi\sb{p}+m\phi\sb{e}+\phi\sb{e}
\Delta^{-1}\big(\abs{\phi\sb{e}}^2+\abs{\phi\sb{p}}^2\big),
\nonumber
\\
\omega\phi\sb{p}
=-i\bm\sigma\!\cdot\!\bm\nabla\phi\sb{e}-m\phi\sb{p}+\phi\sb{p}
\Delta^{-1}\big(\abs{\phi\sb{e}}^2+\abs{\phi\sb{p}}^2\big),
\end{eqnarray}
where
$\bm\sigma\!\cdot\!\bm\nabla=\sum\sb{j=1}\sp{3}\sigma\sb j\p\sb j$,
with
$\sigma\sb j$ the Pauli matrices.
Let $\epsilon>0$
be such that $\epsilon^2=m^2-\omega^2$.
We introduce functions
$\varPhi\sb{e}(y,\epsilon),\ \varPhi\sb{p}(y,\epsilon)\in\C^2$
by the relations
$$
\phi\sb{e}(x,\omega)=\epsilon^2\varPhi\sb{e}(\epsilon x,\epsilon),
\quad
\phi\sb{p}(x,\omega)=\epsilon^3\varPhi\sb{p}(\epsilon x,\epsilon).
$$
Let
$\bm\nabla\sb{y}$,
$\Delta\sb{y}$
be the
gradient
and the Laplacian
with
respect to the coordinates $y=\epsilon x$,
so that
$\ \bm\nabla\sb{x}=\epsilon\bm\nabla\sb{y}$,
$\ \Delta\sb{x}=\epsilon^2\Delta\sb{y}$.
Then equations
\eqref{omega-phi-is-2}
take the form
\begin{equation}\label{sys-phi1}
-\frac{\varPhi\sb{e}}{m+\omega}
=-i
\bm\sigma\!\cdot\!\bm\nabla\sb y
\varPhi\sb{p}+\varPhi\sb{e}
\Delta\sb{y}^{-1}(\abs{\varPhi\sb{e}}^2+\epsilon^2\abs{\varPhi\sb{p}}^2),
\end{equation}
\begin{equation}\label{sys-phi2}
(m+\omega)\varPhi\sb{p}
=-i
\bm\sigma\!\cdot\!\bm\nabla\sb y
\varPhi\sb{e}+\epsilon^2\varPhi\sb{p}
\Delta\sb{y}^{-1}(\abs{\varPhi\sb{e}}^2+\epsilon^2\abs{\varPhi\sb{p}}^2).
\end{equation}
Let $u\in H\sp\infty(\R^3,\R)$
be a spherically symmetric strictly positive smooth solution to
the Choquard equation,
\begin{equation}\label{pc}
-\frac{1}{2m}u
=-\frac{1}{2m}\Delta u+u\Delta^{-1}u^2;
\end{equation}
such a solution exists due to
\cite{MR0471785,MR591299,MR677997}.
Pick a unit vector $\bm{n}\in\C^2$.
Then
$$
\hat\varPhi\sb e=\bm{n}u\in H\sp\infty(\R^3,\C^4),
$$
$$
\hat\varPhi\sb p=
-\frac{1}{2m}
i
\bm\sigma\!\cdot\!\bm\nabla\sb y
\hat\varPhi\sb e\in H\sp\infty(\R^3,\C^4)
$$
is a solution to
\eqref{sys-phi1}, \eqref{sys-phi2}
corresponding to $\epsilon=0$.
By \cite{dm-existence},
the perturbation theory
allows to construct solutions
to \eqref{omega-phi-is}
with $\omega\in(\omega\sb 0,m)$,
with some $\omega\sb 0<m$,
such that
\begin{equation}\label{phi-asymptotics}
\phi\sb\omega(x)
=
\begin{bmatrix}
\phi\sb e(x,\omega)\\\phi\sb p(x,\omega)
\end{bmatrix}
=
\begin{bmatrix}
\epsilon^2\hat\varPhi\sb e(\epsilon x)+o(\epsilon^2)
\\
\epsilon^3\hat\varPhi\sb p(\epsilon x)+o(\epsilon^3)
\end{bmatrix},
\end{equation}
where $\omega$ and $\epsilon$
are related by $\omega=\sqrt{m^2-\epsilon^2}$.

\begin{remark}
Let us mention that among the solitary waves
considered above,
only those with the discrete values
$\omega_{n, \kappa, \charge}\in(\omega\sb 0,m)$
satisfy the constraint (\ref{constrnt2})
(or the constraint (\ref{constrnt2gr})).
These values are parametrized by the number of nodes $n$
of the corresponding solution to the Choquard equation \eqref{pc},
the quantum number $\kappa=\pm 1$,
and the value of charge $\charge$ that enters the constraint (\ref{constrnt2})
(or (\ref{constrnt2gr})).
In the context of the Dirac--Maxwell system,
this pattern is described in detail
in \cite{MR1364144}.
\end{remark}

\subsection{Linear stability of solitary waves}
\label{subsect-stability}

We assume that $\omega\sb 0<m$
is such that
for $\omega\in(\omega\sb 0,m)$
there are solitary wave solutions
$\phi\sb\omega(x)e^{-i\omega t}$
to \eqref{2dd-1}.
Taking the Ansatz
$\zeta(x,t)=(\phi\sb\omega(x)+\rho(x,t))e^{-i\omega t}$,
we derive
the linearization at the solitary wave
$\phi\sb\omega(x)e^{-i\omega t}$:
$$
i\dot\rho
=(
D\sb m
-\omega+\Delta^{-1}\abs{\phi\sb\omega}^2)\rho
+\Delta^{-1}(\rho\sp\ast\phi\sb\omega+\phi\sb\omega\sp\ast\rho)\phi\sb\omega,
$$
where
$D\sb m=-i
\bm\alpha\!\cdot\!\bm\nabla+m\beta$.
We are looking for the eigenvalues
of the linearization operator in the right-hand side.
That is, we substitute
$\rho(x,t)=\xi(x)e^{\lambda t}$,
with $\xi\in L\sp 2(\R^3,\C^4)$,
$\xi\not\equiv 0$,
getting
\begin{equation}\label{it-it}
i\lambda\xi
=(
D\sb m
-\omega+\Delta^{-1}\abs{\phi\sb\omega}^2)\xi
+\phi\sb\omega\Delta^{-1}(\xi\sp\ast\phi\sb\omega+\phi\sb\omega\sp\ast\xi),
\end{equation}
and we would like to know possible values of $\lambda$.
If there is $\Re\lambda>0$
corresponding to $\xi\not\equiv 0$,
then the linearization at a solitary wave is linearly unstable,
and we would expect that the solitary wave is
(``dynamically'') unstable under perturbations of the initial data.

\begin{theorem}\label{theorem-52}
There exists $\omega\sb 1\in [\omega\sb 0,m)$
such that
the ``no-node'' solitary waves with
$\omega\in(\omega\sb 1,m)$ are linearly stable,
so that there are no solution
$\lambda\in\C$, $\xi\in L^2(\R^3,\C^4)$
to \eqref{it-it}
with $\Re\lambda\ne 0$ and $\xi(x)$ not identically zero.
\end{theorem}

\begin{remark}
By the ``no-node'' solitary waves we mean the solutions
\eqref{phi-asymptotics}
constructed from the strictly positive solution
to the Choquard equation \eqref{pc}.
\end{remark}

Let us mention that such ``no-node''
solutions that satisfy the constraint (\ref{constrnt2})
(or the constraint (\ref{constrnt2gr}))
only exist if the value of $\charge^2$ is sufficiently small.
That is why the above results on the existence and stability
may be reformulated as follows:

\begin{theorem}\label{theorem-52-2}
There is $\\charge_0^2>0$
such that
the Dirac--Coulomb system \eqref{2dd-0}
with $\charge^2\in(0,\\charge_0^2)$
has ``no-node'' solitary wave solutions
$\zeta(x,t)=\phi(x)e^{-i\omega t}$
which satisfy
the constraint (\ref{constrnt2}) (or (\ref{constrnt2gr}))
and are linearly stable.
\end{theorem}

\begin{remark}
According to the scaling in the Ansatz (\ref{phi-asymptotics}),
one has
$\int\abs{\phi\sb\omega}^2\,d^3x\sim\epsilon\sim(m-\omega)^{1/2}$,
hence $\charge^2$ and $\omega$
in Theorem~\ref{theorem-52-2}
are related by
$$
\charge^2\sim (m-\omega)^{1/2},
\qquad
\omega\lesssim m.
$$
\end{remark}

We point out that \eqref{it-it}
is $\R$-linear but not $\C$-linear,
because of the presence of $\xi\sp\ast$.
Let us rewrite \eqref{it-it}
in the $\C$-linear form.
For this,
we introduce the following notations:
$$
\bmupxi=\begin{bmatrix}\Re\xi\\\Im\xi\end{bmatrix},
\quad
\bmupphi\sb\omega=\begin{bmatrix}\Re\phi\sb\omega\\\Im\phi\sb\omega\end{bmatrix};
\qquad
\eubJ=\begin{bmatrix}0&I\sb{4}\\-I\sb{4}&0\end{bmatrix},
$$
$$
\bmupalpha\sp j=\begin{bmatrix}
\Re\alpha\sp j&-\Im\alpha\sp j
\\
\Im\alpha\sp j&\Re\alpha\sp j
\end{bmatrix},
\quad
\bmupbeta=\begin{bmatrix}
\Re\beta&-\Im\beta
\\
\Im\beta&\Re\beta
\end{bmatrix}.
$$
Then
\eqref{it-it} can be written as
\begin{equation}\label{it-it-it}
\lambda\bmupxi
=\eubJ\eubL(\omega)\bmupxi,
\end{equation}
where
$$
\eubL(\omega)\bmupxi
=
\big(\eubD\sb m-\omega
+\Delta^{-1}\abs{\bmupphi\sb\omega}^2
\big)\bmupxi
+2\bmupphi\sb\omega\Delta^{-1}(\bmupphi\sb\omega\sp\ast\bmupxi),
$$
$$
\eubD\sb m
=
\sum\limits\sb{j=1}\sp{3}\eubJ\bmupalpha\sp j\p\sb j+\bmupbeta m.
$$
The operators
$\eubD\sb m$ and $\eubL(\omega)$
considered on the domain $H^1(\R^3,\C^8)$
are self-adjoint.

Theorem~\ref{theorem-52}
is the immediate consequence of the following lemma.

\begin{lemma}\label{lemma-lambda-small}
Let
$\omega\sb k\in(0,m)$,
$k\in\N$;
$\omega\sb k\to m$
as $k\to\infty$.
Then there is no sequence
$\lambda\sb k\in\sigma\sb p(\eubJ\eubL(\omega\sb k))$
with
$\Re\lambda\sb k\ne 0$.
\end{lemma}

\begin{proof}
The lemma is proved in several steps,
which we now sketch;
more details will appear in
\cite{dirac-spectrum}.
First, one shows that if there were a sequence
of eigenvalues
$\lambda\sb k\in\sigma\sb p(\eubJ\eubL(\omega\sb k))$
such that $\lim\limits\sb{k\to\infty}\lambda\sb k$
existed,
then we would have
$$
\lim\sb{k\to\infty}\lambda\sb k\subset\{0,\pm 2 m i\}.
$$
The proof of this statement follows from
the fact that in the limit $\omega\to m$,
as $\norm{\phi\sb\omega}\sb{L\sp\infty}\to 0$
(Cf. \eqref{phi-asymptotics}),
the operator
$\eubJ\eubL(\omega)$ turns into $\eubJ(\eubD\sb m-m)$.
According to \cite{MR0320547},
there is the limiting absorption principle
for the free Dirac operator $\eubD\sb m$;
its resolvent, $(\eubD\sb m-z)^{-1}$,
is uniformly bounded
from
$L\sp 2\sb{s}(\R^3,\C^4)$
to $L\sp 2\sb{-s}(\R^3,\C^4)$,
for any $s>1/2$ and
uniformly for $\abs{\Re z}>m+\delta$
(for any fixed $\delta>0$)
and $\Im z\ne 0$.
(Recall that
$L^2\sb{s}(\R^n)=\{u\in L^2\sb{loc}(\R^n)
\sothat
\norm{u}\sb{L^2\sb s}^2
:=
\int\sb{\R^n}(1+x^2)^s\abs{u(x)}^2\,d^n x\}<\infty$.)
This implies that the
resolvent of
$\eubJ\eubL(\omega)$
is bounded in these weighted spaces
outside of the union of $i\R$
with open neighborhoods of ``thresholds'' $\lambda=0$ and $\lambda=\pm 2m i$,
as long as $\omega$ is sufficiently close to $m$.
In turn, this implies that as $\omega\sb k\to m$,
the eigenvalues $\lambda\sb k$ can not accumulate
but to these three threshold points.

Further,
the eigenvalues $\lambda\sb k$ with $\Re\lambda\sb k\ne 0$
can not accumulate to $\pm 2 m i$.
This follows  from the fact that
if $\lambda\sb k\to\lambda\sb b\in i\R\backslash 0$
as $\omega\sb k\to\omega\sb b\in i\R\backslash 0$,
then $\lambda\sb b$ itself has to belong to the point spectrum
of $\eubJ\eubL(\omega\sb b)$
(corresponds to the $L\sp 2$ eigenfunction);
this result is again based on the
limiting absorption principle.
At the same time, there can be no $L\sp 2$ eigenfunctions
of a constant coefficient operator $\eubJ(\eubD\sb m-m)$.

Finally,
one has to study the most involved case
$\lambda\sb k\to 0$.
One first proves that
if $\lambda\sb k\to 0$
and $\Re\lambda\sb k\ne 0$
as $\omega\sb k\to m$,
then necessarily $\lambda\sb k=O(m-\omega\sb k)$.
Then one studies the rescaled equation.
The conclusion is that
the families of eigenvalues for the
linearization of the Dirac--Choquard equation,
$\lambda\sb k\in\sigma\sb p(\eubJ\eubL(\omega\sb k))$,
are deformations
of families of eigenvalues for the linearization of the Choquard equation,
which is a nonrelativistic limit of the Dirac--Choquard equation;
in the context of the nonlinear Dirac equation,
this has been rigorously done in \cite{dirac-spectrum}.
The presence of eigenvalues with nonzero real part
in the linearization of Choquard equation
is controlled by the
Vakhitov--Kolokolov stability criterion \cite{VaKo};
for the linearization at
no-node solutions,
this criterion prohibits existence
of such eigenvalues.
This finishes the proof of the lemma.
\end{proof}

We reproduce the Vakhitov--Kolokolov stability criterion \cite{VaKo}
in application to the Choquard equation in the Appendix
(see Lemma~\ref{lemma-61} below).

\section{Conclusions}
\label{sectconcl}

In the present paper we considered solitary waves
in the system of Dirac fermions interacting via the Coulomb
attraction from
both Physics and Mathematics viewpoints.
On the physical side the solitary waves describe polarons
that may appear in two situations.
The first one corresponds to certain condensed matter systems
in which massive Dirac fermions interact with optical phonons.
The polarons in the system of true relativistic Dirac fermions
interacting with gravity in the Newtonian limit
are also described by the above-mentioned solitary waves.

A possible application of our construction
for the gravitational case is related to the situation
when the gravitational interaction between elementary particles is strong enough.
The corresponding problem may only appear in the models that rely on quantum gravity.
It is worth mentioning here that the role of the gravitational interaction
may be played by the emergent gravity \cite{MR1989766,unruh2010quantum-volovik}
with the scale much lower than the Planck mass.
Such models may be relevant
for the description of the TeV-scale physics \cite{2012arXiv1209.0204V}.

On the mathematical side we develop analytical methods
for the investigation of solitary waves.
These methods are based on the   observation that these localized solutions
are obtained as a bifurcation from the solitary waves
of the Choquard equation.
Basing on this approach, we demonstrate that the no-node gap solitons
for sufficiently small values of $\charge$ are linearly stable.

It is worth mentioning that the solitary waves similar to the considered in the
present paper may also exist in two-dimensional systems like the boundary of
the topological insulators or the graphene. This may occur if the interaction
between the electrons of the 2D system with the balk phonons (with the
substrate phonons in the case of graphene) is strong enough. We postpone the
consideration of the corresponding 2D solitary waves to future publications.

\medskip

\noindent
ACKNOWLEDGMENTS.
The authors would like to
kindly acknowledge the private communication with G.E.\,Volovik,
who prompted to consider the relation of the polaron problem
to the solitary waves in the Dirac--Coulomb system.
A.C. benefited from discussions
with Gregory Berkolaiko, Nabile Boussa\"id, and Boris Vainberg.

The work of M.A.Z. was partially supported by RFBR grant 11-02-01227,  by
the Federal Special-Purpose Programme ``Human Capital''
of the Russian Ministry of Science and Education, and by the Federal
Special-Purpose Programme 07.514.12.4028.

The authors are grateful to August Krueger for his advice and corrections.

\appendix

\section{
Vakhitov--Kolokolov criterion
for the Choquard equation}

We consider the Choquard equation,
\begin{equation}\label{choquard}
i\p\sb t\zeta=-\frac{1}{2m}\Delta\zeta
+m\zeta
+\zeta\Delta^{-1}\abs{\zeta}^2,
\end{equation}
where $\zeta(x,t)\in\C^3$
and $x\in\R^3$.
We are interested in the solitary wave solutions
$
\zeta(x,t)=u\sb\omega(x)e^{-i\omega t},
$
$\omega\in\R$;
$u\sb\omega$
satisfies
\begin{equation}\label{choquard-st}
-(m-\omega)u\sb\omega
=-\frac{1}{2m}\Delta u\sb\omega
+u\sb\omega\Delta^{-1}u\sb\omega^2.
\end{equation}
%
Given a solution to \eqref{pc},
$$
-\frac{1}{2m}u
=-\frac{1}{2m}\Delta u+u\Delta^{-1}u^2,
$$
then, for any $\omega<m$,
the profiles
\begin{equation}\label{scal}
u\sb\omega(x)=2m(m-\omega)
\,u\big(x\sqrt{2m(m-\omega)}\,\big)
\end{equation}
correspond to a family of solitary wave solutions to
\eqref{choquard-st}.
Note that this scaling is the same as that of
$\phi\sb e$ in
\eqref{phi-asymptotics}.

\begin{lemma}\label{lemma-61}
For $\omega<m$,
the no-node solitary wave solutions
$\zeta\sb\omega(x,t)=u\sb\omega(x)e^{-i\omega t}$
to the Choquard system
are linearly stable.
\end{lemma}

The linear stability of no-node solitary waves
of the Choquard equation
follows from the Vakhitov--Kolokolov stability criterion
\cite{VaKo}
which is applicable to systems of the Schr\"odinger type.
It also follows from \cite{MR677997}
(where the orbital stability of these solitary waves is proved).
Let us sketch the argument.
First, we notice that, by \eqref{scal},
the charge of the solitary wave
$u\sb\omega(x)e^{-i\omega t}$ is given by
$$
Q(\phi_\omega) =\int\sb{\R^3}\abs{\phi_\omega(x)}^2\,d^3x
\sim
(m-\omega)^{1/2},
\qquad
\omega\lesssim m;
$$
therefore,
\begin{equation}\label{neg}
\frac{d Q(u\sb\omega)}{d\omega}<0,
\qquad
\omega<m.
\end{equation}
We consider the solution to the
Choquard equation in the form
of a perturbed solitary wave,
$$
\zeta(x,t)=(u\sb\omega(x)+R(x,t)+i S(x,t))e^{-i\omega t},
$$
with $R$, $S$ real-valued.
The linearized equation on $R$, $S$
is given by
\begin{equation}
\p\sb t
\begin{bmatrix}R\\S\end{bmatrix}
=
\eub{j}\eub{l}(\omega)
\begin{bmatrix}R\\S\end{bmatrix},
\end{equation}
where
\begin{equation}
\eub{j}
=
\begin{bmatrix}0&1\\-1&0\end{bmatrix},
\qquad
\eub{l}(\omega)=
\begin{bmatrix}L\sb 1(\omega)&0\\0&L\sb 0(\omega)\end{bmatrix},
\end{equation}
and
$$
L\sb 0(\omega)=-\frac{1}{2m}\Delta+m-\omega
+\Delta^{-1}u\sb\omega^2,
$$
$$
L\sb 1(\omega)=L\sb 0(\omega)+2\Delta^{-1}(u\sb\omega\,\cdot\,)u\sb\omega.
$$
Both operators $L\sb 0(\omega)$ and $L\sb 1(\omega)$ are self-adjoint,
with
$\sigma\sb{ess}(L\sb 0(\omega))
=\sigma\sb{ess}(L\sb 1(\omega))
=[m-\omega,+\infty)$.
Clearly, $L\sb 0(\omega) u\sb\omega=0$,
with $0\in\sigma\sb d(L\sb 0)$
an eigenvalue
corresponding to a positive eigenfunction
$u\sb\omega$;
it follows that $0$ is a simple eigenvalue of $L\sb 0$,
with the rest of the spectrum separated from zero.
Taking the derivatives of the equality
$L\sb 0(\omega)u\sb\omega=0$
with respect to $x\sb j$ and $\omega$,
we get:
\begin{equation}\label{wg}
L\sb 1(\omega)\p\sb{x\sb j}u\sb\omega=0,
\qquad
L\sb 1(\omega)\p\sb\omega u\sb\omega=u\sb\omega.
\end{equation}
The first relation shows that $\lambda\sb 1=0$
is the point eigenvalue of $L\sb 1(\omega)$,
and since $\p\sb{x\sb j}u\sb\omega$
vanishes on a hyperplane $x\sb j=0$,
there is one negative eigenvalue $\lambda\sb 0<0$
of $L\sb 1(\omega)$.

Now we may determine the spectrum of
$\eub{j}\eub{l}(\omega)
=\begin{bmatrix}0&L\sb 0(\omega)\\-L\sb 1(\omega)&0\end{bmatrix}$.
We closely follow \cite{VaKo}.
If $\begin{bmatrix}R\\S\end{bmatrix}$ is an eigenfunction
corresponding to the eigenvalue $\lambda\in\C$,
then $-\lambda^2 R=L\sb 0 L\sb 1 R$.
If $\lambda\ne 0$,
then one concludes that
$R$ is orthogonal to $\ker L\sb 0=\mathop{\rm Span}(u\sb\omega)$,
hence we can apply $L\sb 0^{-1}$;
taking then the inner product with $u\sb\omega$,
we have:
\begin{equation}\label{minus-lambda}
-\lambda^2
\langle R,L\sb 0^{-1}R\rangle
=
\langle R,L\sb 1 R\rangle.
\end{equation}
With $L\sb 0$, $L\sb 1$
being self-adjoint,
this relation implies that $\lambda^2\in\R$.
Since $L\sb 0(\omega)$
is non-negative and $R\perp\ker L\sb 0$,
one has
$\langle u\sb\omega,L\sb 0\sp{-1}u\sb\omega\rangle>0$.
The solution to
$$
\mu:=\inf\big\{\langle R,L\sb 1(\omega),R\rangle\sothat
\norm{R}=1,\ \langle u\sb\omega,R\rangle=0\big\}
$$
satisfies
$L\sb 1(\omega) R=\mu R+\nu u\sb\omega$,
where $\mu$, $\nu\in\R$ play the role of the Lagrange multipliers.
Due to the condition $\langle u\sb\omega,R\rangle=0$,
$\mu$ delivers the zero value to the function
$$
f(z)=\langle u\sb\omega,(L\sb 1(\omega)-z)^{-1}u\sb\omega\rangle,
\qquad
z\in\rho(L\sb 1(\omega)),
$$
with $\rho(L\sb 1)$
denoting the resolvent set of $L\sb 1$.
Since $\ker L\sb 1$ is spanned by
$\p\sb j u\sb\omega$
and therefore is orthogonal to $u\sb\omega$,
we can extend $f(z)$ to
$z\in(\lambda\sb 0,\lambda\sb 2)$,
where
$\lambda\sb 0=\inf \sigma(L\sb 1(\omega))<0$
and $\lambda\sb 2$ is the smallest positive eigenvalue of $L\sb 1$
in the interval $(0,m-\omega)$
(or the edge of the essential spectrum, $\lambda=m-\omega$).
We need to know whether $\mu$ is positive or negative.
Since $f'(z)>0$,
the sign of $\mu$ is opposite to
$$
f(0)
=\langle u\sb\omega,L\sb 1(\omega)\sp{-1}u\sb\omega\rangle
=\langle u\sb\omega,\p\sb\omega u\sb\omega\rangle
=\frac{\p\sb\omega Q(u\sb\omega)}{2}.
$$
In the second equality,
we used the second relation from \eqref{wg}.
  From \eqref{neg},
we conclude that
$f(0)<0$;
thus, $\mu>0$.
By \eqref{minus-lambda},
$\lambda^2\le 0$,
leading to
$\spec(\eub{j}\eub{l})\subset i\R$.

\medskip

This shows that there are no families of eigenvalues of
$\eub{j}\eub{l}(\omega)$ with nonzero real part
bifurcating from $\lambda=0$ at $\omega=m$.
Since bifurcations of eigenvalues from $\lambda=0$
for the linearizations
of the Choquard equation
and
the Dirac--Choquard equation \eqref{2dd-1}
(which is equivalent to the Dirac--Coulomb system)
have the same asymptotics as $\omega\to m$,
we conclude that
neither are there families of eigenvalues
of $\eubJ\eubL(\omega)$ with nonzero real part.

\begin{remark}
The rigorous proof of linear stability of solitary wave solutions
to the Dirac--Choquard equation
requires a more detailed analysis of the spectrum of $\eub{j}\eub{l}(\omega)$.
Namely, one needs to know whether
there are resonances or embedded eigenvalues.
Theoretically, resonances or
embedded eigenvalues of higher algebraic multiplicity
could bifurcate off the imaginary axis into the complex domain,
yielding a family of eigenvalues
$\lambda\sb k\in\sigma\sb{p}(\eubJ\eubL(\omega\sb k))$
with $\omega\sb k\to m$, $\lambda\sb k=O(m-\omega\sb k)$,
$\Re\lambda\sb k\ne 0$
(and resulting in the instability of Dirac--Choquard solitary waves),
although we expect that generically this does not happen.
\end{remark}

\notyet{

\subsection{Numerical construction
of solitary wave solutions
to the Choquard equation}

Numerically, the spherically symmetric solutions
to the Choquard equation
could be obtained as follows.
We take
$$
m=1.
$$
We denote $W(r)=-v(r)+E$
and rewrite
\eqref{lambda-phi}
as the first order system:
$$
\Pi(r)=u',
\qquad
0=-\frac 1 2(\Pi'+2\Pi)-W u,
\qquad
Z(r)=W'(r),
\qquad
-(Z'+2Z)=u^2.
$$
Without loss of generality,
we can assume that $u(0)=1$.
The initial data are
$u(0)=1$,
$\Pi(0)=u'(0)=0$,
$Z(0)=W'(0)=0$.
The value of $W(0)$ is being chosen so that
$\lim\sb{r\to\infty}u(r)=0$.
The solutions were obtained
via the shooting method,
on the interval $0\le r\le R:\approx 34\approx \lg(1/\hbar)$.
Once we find the solution $u$,
we compute its approximate total charge,
$Q=\int_0^R u^2(r)r^2\,dr$.
We take $E<0$ so that
$E+Q/r$
matches $W(r)$ for large $r$;
that is, we take
\begin{equation}\label{eis}
E=W(R)-\frac Q R.
\end{equation}
See Figure~\ref{ns-nodes}.
Here are the first several solutions.

\begin{verse}
\noindent
No nodes ($N=0$):
$u(0) =  1$,
$W(0) \approx  0.650$,
$Q_0 \approx  2.066$,
$E_0 \approx -0.692$.

\noindent
One node ($N=1$):
$u(0) =  1$,
$W(0) \approx  0.856$,
$Q_1 \approx  4.587$,
$E_1 \approx -0.648$.

\noindent
Two nodes ($N=2$):
$u(0) =  1$,
$W(0) \approx  0.950$,
$Q_2 \approx  7.098$,
$E_2 \approx -0.631$.

\noindent
Three nodes ($N=3$):
$u(0) =  1$,
$W(0) \approx  1.001$,
$Q_3 \approx  9.592$,
$E_3 \approx -0.621$.

\noindent
Four nodes $(N=4)$:
$u(0) =  1$,
$W(0) \approx  1.053$,
$Q_4 \approx  12.077$,
$E_4 \approx -0.614$.
\end{verse}


\begin{figure}[htbp]
\begin{center}
includegraphics[width=15cm,height=4cm]{dirac-maxwell-3d-00}

includegraphics[width=15cm,height=4cm]{dirac-maxwell-3d-01}

includegraphics[width=15cm,height=4cm]{dirac-maxwell-3d-02}

includegraphics[width=15cm,height=4cm]{dirac-maxwell-3d-03}

includegraphics[width=15cm,height=4cm]{dirac-maxwell-3d-04}
\end{center}

\caption{\footnotesize
Solution pairs $(u(r),W(r))$,
with
$u$ (bold lines)
having the number of nodes $N=0,\,1,\,2,\,3,\,4$.
The marks ``$+$''
denote the shape of the function
$E+Q/r$,
which is to coincide with the asymptotic
behavior of the potential $W(r)$
(plotted by thin lines)
for large $r$.
Matching these asymptotics
allows us to find the value of $E$
(Cf. \eqref{eis}).
}
\label{ns-nodes}
\end{figure}

}

\bibliography{all,zubkov}
\bibliographystyle{iopart-numa}

\def\cprime{$'$} \def\polhk#1{\setbox0=\hbox{#1}{\ooalign{\hidewidth
  \lower1.5ex\hbox{`}\hidewidth\crcr\unhbox0}}}
\providecommand{\newblock}{}

\end{document}